\documentclass[a4paper,11pt,fleqn]{article}
\usepackage[margin=1.25in]{geometry}
\usepackage{fullpage}
\usepackage{amssymb,amsfonts,amsmath,amsthm}
\usepackage{mathrsfs}
\usepackage{amscd}
\usepackage{xcolor}
\usepackage{array}
\usepackage[linesnumbered,ruled]{algorithm2e}
\usepackage{caption}
\usepackage{subcaption}
\usepackage{enumerate}
\usepackage{doi}

\title{Optimal Spectral-Norm Approximate Minimization of Weighted Finite Automata}
\usepackage{authblk}
\date{}

\author[1]{Borja Balle}
\author[2,3]{Clara Lacroce\footnote{Corresponding author: clara.lacroce@mail.mcgill.ca}}
\author[2,3]{Prakash Panangaden}
\author[2,3]{Doina Precup}
\author[3,4]{Guillaume Rabusseau\footnote{The names of the authors appear in alphabetical order}}

\affil[1]{DeepMind, London, United Kingdom.}
\affil[2]{School of Computer Science, McGill University, Montr\'eal, Canada}
\affil[3]{Mila, Montr\'eal, Canada}
\affil[4]{DIRO, Universit\'e de Montr\'eal, Montr\'eal, Canada}

\newcommand{\mat}[1]{\mathbf{#1}}
\newcommand{\norm}[1]{\|#1\|}
\newcommand{\rank}{\operatorname{rank}}

\theoremstyle{plain}
\newtheorem{definition}{Definition}[section]
\newtheorem{corollary}{Corollary}
\newtheorem{theorem}{Theorem}[section]
 
\newtheorem{lemma}[theorem]{Lemma}

\theoremstyle{definition}
\newtheorem{remark}{Remark}
\newtheorem{example}[theorem]{Example}

\newcommand{\C}{\mathbb{C}}
\newcommand{\N}{\mathbb{N}}
\newcommand{\Z}{\mathbb{Z}}
\newcommand{\R}{\mathbb{R}}
\newcommand{\A}{\mat{A}}
\renewcommand{\H}{\mat{H}}

\newcommand{\mP}{\mat{P}}
\newcommand{\mQ}{\mat{Q}}
\newcommand{\mT}{\mat{T}}

\newcommand{\balpha}{\boldsymbol{\alpha}}
\newcommand{\bbeta}{\boldsymbol{\beta}}
\newcommand{\wfa}{\langle \balpha , \{\A_a\},  \bbeta \rangle}
\newcommand{\wa}{\langle \balpha , \A,  \bbeta \rangle}

\begin{document}

\maketitle

\begin{abstract}

We address the approximate minimization problem for weighted finite automata (WFAs) with weights in $\R$, over a one-letter alphabet: to compute the best possible approximation of a WFA given a bound on the number of states. This work is grounded in Adamyan-Arov-Krein Approximation theory, a remarkable collection of results on the approximation of Hankel operators. In addition to its intrinsic mathematical relevance, this theory has proven to be very effective for model reduction. We adapt these results to the framework of weighted automata over a one-letter alphabet. We provide theoretical guarantees and bounds on the quality of the approximation in the spectral and $\ell^2$ norm. We develop an algorithm that, based on the properties of Hankel operators, returns the optimal approximation in the spectral norm. 
\end{abstract}

\section{Introduction}

    Weighted finite automata (WFAs) are an expressive class of models representing functions defined over sequences. The \emph{approximate minimization problem} is concerned with finding an automaton that approximates the behaviour of a given minimal WFA, while being smaller in size. This second automaton recognizes a different language, and the objective is to minimize the approximation error~\cite{Balle15,Balle19}. Approximate minimization becomes particularly useful in the context of spectral learning algorithms~\cite{BaillySpectral,BalleCLQ14,ballewill,Hsu}. When applied to a learning task, such algorithms can be viewed as working in two steps. First, they compute a minimal WFA that explains the training data exactly. Then, they obtain a model that generalizes to the unseen data by producing a smaller approximation to the minimal WFA. It is not just a question of saving space by having a smaller state space; the exact machine will \emph{overfit} the data and generalize poorly. To obtain accurate results it is crucial to guess correctly the size of the minimal WFA, in particular when the data is generated by this type of machine.
    
    The minimization task is greatly shaped by the way we decide to measure the approximation error. It is thus natural to wonder if there are norms that are preferable to others. We believe that the chosen norm should be computationally reasonable to minimize. For instance, the distance between WFAs can be computed using a metric based on bisimulation~\cite{Balle17}. While exploring this approach could still be interesting, the fact that this metric is hard to compute makes it unsuitable for our purposes. Moreover, this metric is specifically designed for WFAs, so it is not directly applicable to other models dealing with sequential data. We think that being transferable is a second important feature for the chosen norm. In fact, being able to compare different classes of models is desirable for future applications of this method. For example, one can think of the burgeoning literature on approximating Recurrent Neural Networks (RNNs) using WFAs, where the objective is to extract from a trained RNN an automaton that accurately mimics its behaviour~\cite{Rabusseau19,WeissWFA19,Takamasa,Ayache2018,eyraud2020}. With this in mind, we think that it is preferable to consider a norm defined on the input-output function -- or the Hankel matrix --  rather than the parameters of the specific model considered. Finally, it is important to choose a norm that can be computed accurately. The spectral norm seems to be a good candidate for the task. In particular, it allows us to exploit the work of Adamyan, Arov and Krein which has come to be known as AAK theory~\cite{AAK71}: a series of results connecting the theory of complex functions on the unit circle to Hankel matrices, a mathematical object representing functions defined over sequences. The core of this theory provides us with theoretical guarantees for the exact computation of the spectral norm of the error, and a method to construct the optimal approximation.
    
    We summarize our main contributions:
    \begin{itemize}
        \item We use AAK theory to study the approximate minimization problem of WFAs. To connect those areas, we establish a correspondence between the parameters of a WFA and the coefficients of a complex function on the unit circle. To the best of our knowledge, this paper represents the first attempt to apply AAK theory to WFAs. 
        \item We present a theoretical analysis of the optimal spectral-norm approximate minimization problem for WFAs, based on their connection with finite-rank infinite Hankel matrices. We provide a closed form solution for real weighted automata over a one-letter alphabet $A=\wa$, under the assumption $\rho(\A)<1$ on the spectral radius. We bound the approximation error, both in terms of the Hankel matrix (spectral norm) and of the rational function computed by the WFA ($\ell^2$ norm).
        \item We propose a self-contained algorithm that returns the unique optimal spectral-norm approximation of a given size.
        \item We tighten the connection, made in~\cite{Balle19}, between WFAs and discrete dynamical systems, by adapting some of the control theory concepts, like the \emph{allpass system}~\cite{Glover}, to the formalism of WFAs.
    \end{itemize}

\section{Background}
\subsection{Preliminaries}

    We denote with $\N$, $\Z$ and $\R$ the set of natural, integers and real numbers, respectively. We use bold letters for vectors and matrices; all vectors considered are column vectors. We denote with $\mat{1}$ the identity matrix, specifying its dimension only when not clear from the context. We refer to the $i$-th row and the $j$-th column of $\mat{M}$ by $\mat{M}(i,:)$ and $\mat{M}(:,j)$. Given a matrix $\mat{M}\in \R^{p\times q} $ of rank $n$, a \emph{rank factorization} is a factorization $\mat{M}=\mP\mQ$, where $\mP \in \R^{p\times n}$, $\mQ \in \R^{n\times q}$ and $\rank(\mP)=\rank(\mQ)=n$. Let $\mat{M} \in \R^{p \times q}$ of rank $n$, the compact \emph{singular value decomposition} SVD of $\mat{M}$ is the factorization $\mat{M}=\mat{U}\mat{D}\mat{V}^{\top}$, where $\mat{U}\in \R^{p\times n}$, $\mat{D}\in \R^{n\times n}$, $\mat{V}\in \R^{q \times n}$ are such that $\mat{U}^{\top}\mat{U}=\mat{V}^{\top}\mat{V}=\mat{1}$, and $\mat{D}$ is a diagonal matrix.
    The columns of $\mat{U}$ and $\mat{V}$ are called left and right \emph{singular vectors}, while the entries of $\mat{D}$ are the \emph{singular values}. The \emph{Moore-Penrose pseudo-inverse} $\mat{M}^+$ of $\mat{M}$ is the unique matrix such that $\mat{M}\mat{M}^+\mat{M}=\mat{M}$, $\mat{M}^+\mat{M}\mat{M}^+=\mat{M}^+$, with $\mat{M}^+\mat{M}$ and $\mat{M}\mat{M}^+$ Hermitian~\cite{Zhu}.
    The \emph{spectral radius} $\rho(\mat{M})$ of a matrix $\mat{M}$ is the largest modulus among its eigenvalues. 

    A \emph{Hilbert space} is a complete normed vector space where the norm arises from an inner product. A linear operator  $T: X \rightarrow Y$ between Hilbert spaces is \emph{bounded} if it has finite operator norm, \emph{i.e.} $\norm{T}_{op} = \sup_{\norm{g}_X\leq 1}\norm{Tg}_Y < \infty$. We denote by $\mat{T}$ the (infinite) matrix associated with $T$ by some (canonical) orthonormal basis on $H$. An operator is \emph{compact} if the image of the unit ball in $X$ is relatively compact. Given Hilbert spaces $X, Y$ and a compact operator $T:X \rightarrow Y$, we denote its adjoint by $T^*$. The \emph{singular numbers} $\{\sigma_n\}_{n \geq 0}$ of $T$ are the square roots of the eigenvalues of the self-adjoint operator $T^* T$, arranged in decreasing order.
    A $\sigma$-\emph{Schmidt pair} $\{\boldsymbol{\xi}, \boldsymbol{\eta}\}$ for $T$ is a couple of norm $1$ vectors such that: $\mat{T}\boldsymbol{\xi}=\sigma \boldsymbol{\eta}$ and $\mat{T}^*\boldsymbol{\eta}= \sigma\boldsymbol{\xi}$.
    The Hilbert-Schmidt decomposition provides a generalization of the compact SVD for the infinite matrix of a compact operator $T$ using singular numbers and orthonormal Schmidt pairs: $\mT\mat{x}=\sum_{n\geq0}\sigma_n\langle\mat{x},\boldsymbol{\xi}_n \rangle \boldsymbol{\eta}_k$~\cite{Zhu}. The \emph{spectral norm} $\norm{\mat{T}}$ of the matrix representing the operator $T$ is the largest singular number of $T$. Note that the spectral norm of $\mat{T}$ corresponds to the operator norm of $T$.

    Let $\ell^2$ be the Hilbert space of square-summable sequences over $\Sigma^*$, with norm $\norm{f}_2^2=\sum_{x\in\Sigma^*}|f(x)|^2$ and inner product $\langle f, g \rangle = \sum_{x \in \Sigma^*}f(x)g(x)$ for $f,g \in \R^{\Sigma^*}$. Let $\mathbb{T}=\{z\in \C: |z|=1\}$ be the complex unit circle, $\mathbb{D}=\{z\in \C: |z|<1\}$ the (open) complex unit disc. Let $1<p< \infty$, $\mathcal{L}^p(\mathbb{T})$ be the space of measurable functions on $\mathbb{T}$ for which the $p$-th power of the absolute value is Lebesgue integrable. For $p=\infty$, we denote with $\mathcal{L}^{\infty}(\mathbb{T})$ the space of measurable functions that are bounded, with norm $\norm{f}_{\infty}=\sup\{|f(x)|: x\in\mathbb{T}\}$.
  
\subsection{Weighted Finite Automata}
    
    Let $\Sigma$ be a fixed finite alphabet, $\Sigma^*$ the set of all finite strings with symbols in $\Sigma$. We use $\varepsilon$ to denote the empty string. Given $p,s \in \Sigma$, we denote with $ps$ their concatenation.
    
    A \emph{weighted finite automaton} (WFA) of $n$ states over $\Sigma$ is a tuple $A = \wfa$, where $\balpha,$ $\bbeta \in \R^n$ are the vector of initial and final weights, respectively, and $\A_a \in \R^{n \times n}$ is the matrix containing the transition weights associated with each symbol $a$. Every WFA $A$ with real weights realizes a function $f_A : \Sigma^* \to \R$, \emph{i.e.} given a string  $x = x_1 \cdots x_t \in \Sigma^*$, it returns $f_A(x) = \balpha ^\top \A_{x_1} \cdots \A_{x_t} \bbeta = \balpha ^\top \A_x \bbeta$. A function $f : \Sigma^* \to \R$ is called \emph{rational} if there exists a WFA $A$ that realizes it. The \emph{rank} of the function is the size of the smallest WFA realizing $f$. Given $f : \Sigma^* \to \R$, we can consider a matrix $\H_f \in \R^{\Sigma^* \times \Sigma^*}$ having rows and columns indexed by strings and defined by $\H_f(p,s) = f(ps)$ for $p, s \in \Sigma^*$.
  
    \begin{definition}
        A (bi-infinite) matrix $\H \in \R^{\Sigma^* \times \Sigma^*}$ is \textbf{Hankel} if for all $p, p', s, s' \in \Sigma^*$ such that $p s = p' s'$, we have $\H(p,s) = \H(p',s')$. Given a Hankel matrix $\H \in \R^{\Sigma^* \times \Sigma^*}$, there exists a unique function $f : \Sigma^* \to \R$ such that $\H_f = \H$.
    \end{definition}    
    	
    \begin{theorem}[\cite{CP71,Fli}]\label{fliess}
        A function $f:\Sigma^* \to \R$ can be realized by a WFA if and only if $\H_f$ has finite rank $n$. In that case, $n$ is the minimal number of states of any WFA $A$ such that $f = f_A$.
    \end{theorem}
    
    Given a WFA $A = \wfa $, the \emph{forward matrix} of $A$ is the infinite matrix $\mat{F}_A \in \R^{\Sigma^* \times n}$ given by $\mat{F}_A(p,:) = \balpha ^\top \A_p$ for any $p \in \Sigma^*$, while the \emph{backward matrix} of $A$ is $\mat{B}_A \in \R^{\Sigma^* \times n}$, given by $\mat{B}_A(s,:) = (\A_s \bbeta)^\top$ for any $s \in \Sigma^*$. Let $\H_f$ be the Hankel matrix of $f$, its forward-backward (FB) factorization is: $\H_f = \mat{F} \mat{B}^\top$. A WFA with $n$ states is \emph{reachable} if $\rank(\mat{F}_A)=n$, while it is \emph{observable} if $\rank(\mat{B}_A)=n$. A WFA is \emph{minimal} if it is reachable and observable. If $A$ is minimal, the FB factorization is a rank factorization~\cite{BalleCLQ14}.   
    
   We recall the definition of the singular value automaton, a canonical form for WFAs~\cite{Balle15}.
    \begin{definition}
	    Let $f:\Sigma^*\rightarrow \R$ be a rational function and suppose $\H_f$ admits an SVD, $\H_f = \mat{U} \mat{D} \mat{V}^{\top}$. A \textbf{singular value automaton} (SVA) for $f$ is the minimal WFA $A$ realizing $f$ such that $\mat{F}_A=\mat{U} \mat{D}^{1/2}$ and $\mat{B}_A=\mat{V}\mat{D}^{1/2}$.
    \end{definition}
    The SVA can be computed with an efficient algorithm relying on the following matrices~\cite{Balle19}.

    \begin{definition}
        Let $f:\Sigma^*\rightarrow \R$ be a rational function, $\H_f = \mat{F} \mat{B}^\top$ a FB factorization. If the matrices $\mP=\mat{F}^\top \mat{F}$ and $\mQ=\mat{B}^\top \mat{B}$ are well defined, we call $\mP$ the \textbf{reachability Gramian} and $\mQ$ the \textbf{observability Gramian}.
    \end{definition}
    Note that if $A$ is an SVA, then the Gramians associated with its FB factorization satisfy $\mP_A = \mQ_A = \mat{D}$, where $\mat{D}$ is the matrix of singular values of its Hankel matrix.
    The Gramians can alternatively be characterized (and computed~\cite{Balle19}) using fixed point equations, corresponding to Lyapunov equations when $|\Sigma|=1$~\cite{lyapunov}.
    
    \begin{theorem}
        Let $|\Sigma|=1$, $A= \langle \balpha, \A, \bbeta\rangle$ a WFA with $n$ states and well-defined Gramians $\mP$, $\mQ$. Then $X=\mP$ and $Y=\mQ$ solve:
    \begin{align}
       &X-\A X\A^{\top}=\bbeta\bbeta^{\top},\\
       &Y-\A^{\top}Y\A=\balpha\balpha^{\top}.
    \end{align}
    \end{theorem}
    
    Finally, we recall the following definition.
    
    \begin{definition}
        A WFA $A= \wfa$ is a \textbf{generative probabilistic automaton} (GPA) if $f_A(x)\geq 0$ for every $x$, and $\sum_{x\in\Sigma^*}f_A(x)=1$, \emph{i.e.} if $f_A$ computes a probability distribution over $\Sigma^*$.
    \end{definition}
    
    \begin{example}\label{example1}
        Let $|\Sigma|=1$, $\Sigma=\{x\}$. The WFA $A= \langle \balpha, \A, \bbeta\rangle$, with: 
        \begin{equation*}
            \A= \begin{pmatrix}
                    0 & \frac{1}{2}    \\
                    \frac{1}{2}  & 0  
                \end{pmatrix} , \quad
            \balpha= \begin{pmatrix}
                \frac{\sqrt{3}}{2}     \\
                0  
                \end{pmatrix}, \quad
            \bbeta= \begin{pmatrix}
                \frac{\sqrt{3}}{2}     \\
                0  
                \end{pmatrix},
        \end{equation*}
        is a GPA. Ideed, $f_A(x)\geq 0$ and $\sum_{x\in\Sigma^*}f_A(x)=1$, since the rational function is:
        \begin{equation*}
            f_A(x\cdots x)=f_A(k)= \balpha^{\top}\A^k\bbeta= \begin{cases}
                                                0   &\text{if $k$ is odd} \\
                                                \frac{3}{4}2^{-k}   &\text{if $k$ is even}
                                                \end{cases}
        \end{equation*}
       where $k$ corresponds to the string where $x$ is repeated $k$-times.
        We remark that $A$ is minimal and in its SVA form, with Gramians $\mP=\mQ= \begin{pmatrix}
                    \frac{4}{5}  & 0  \\
                    0            & \frac{1}{5} 
                \end{pmatrix}$, and $f_A$ has rank $2$. Finally, the corresponding Hankel matrix, also of rank $2$, is:
        \begin{equation}
            \H=\begin{pmatrix} f_A(0) & f_A(1) & f_A(2) & \dots \\
                               f_A(1) & f_A(2) & f_A(3) &\dots \\
                               f_A(2) & f_A(3) & f_A(4) &\dots \\
                                \vdots& \vdots &\vdots&\ddots
            \end{pmatrix}=
            \begin{pmatrix} \frac{3}{4} & 0 & \frac{3}{16} & \dots \\
                               0 & \frac{3}{16} & 0 &\dots \\
                               \frac{3}{16} & 0 & \frac{3}{64} &\dots \\
                                \vdots& \vdots &\vdots&\ddots
            \end{pmatrix}.
        \end{equation} 
    \end{example}

\subsection{AAK Theory}\label{aak-section}

    Theorem~\ref{fliess} provides us with a way to associate a minimal WFA $A$ with $n$ states to a Hankel matrix $\H$ of rank $n$. The approach we propose to approximate $A$ is to find the WFA corresponding to the matrix that minimizes $\H$ in the spectral norm. We recall the fundamental result of Schmidt, Eckart, Young and Mirsky~\cite{Eckart}. 
    \begin{theorem}[\cite{Eckart}]\label{thm:eckart}
        Let $\H$ be a Hankel matrix corresponding to a compact Hankel operator of rank $n$, and let $\sigma_0 \geq \dots \geq \sigma_{n-1}>0$ be its singular numbers. Then, if $\mat{R}$ is a matrix of rank $k$, we have:
        \begin{equation}
            \norm{\H - \mat{R}}\geq \sigma_k.
        \end{equation}
        The equality is attained when $\mat{R}$ corresponds to the truncated SVD of $\H$. 
    \end{theorem}
    Note that a low-rank approximation obtained by truncating the SVD is not in general a Hankel matrix. This is problematic, since $\mat{G}$ needs to be Hankel in order to be the matrix of a WFA. Surprisingly, we can obtain a result comparable to the one of Theorem~\ref{thm:eckart} while preserving the Hankel property. This is the possible thanks to a theory of optimal approximation called AAK theory~\cite{AAK71}. To apply this theory, we will need to rewrite the approximation problem in functional analysis terms. First, we will associate a linear operator to the Hankel matrix. Then, we will use Fourier analysis to reformulate the problem in a function space. A comprehensive presentation of the concepts recalled in this section can be found in~\cite{Nikolski,Peller,Meinguet}.
    
    
    Let $f:\Sigma^* \rightarrow \R$ be a rational function, we interpret the corresponding Hankel matrix $\H_f$ as the expression of a linear (Hankel) operator $H_f:\ell^2 \rightarrow \ell^2$ in terms of the canonical basis. We recall that a Hankel operator $H_f$ is bounded if and only if $f\in \ell^2$~\cite{Balle19}. This property, together with the fact that we only consider finite rank operators (corresponding to the Hankel matrices of WFAs), is sufficient to guarantee compactness. 

    To introduce AAK theory, we need to consider a second realization of Hankel operators on complex spaces. Since in this paper we work with two classes of functions -- functions over sequences and complex functions -- to avoid any confusion we will make explicit the dependence on the complex variable $z=e^{it}$. We start by recalling a few fundamental definitions from the theory of complex functions. Note that a function $\phi(z) \in \mathcal{L}^2(\mathbb{T})$ can be represented, using the orthonormal basis $\{z^n\}_{n \in \Z}$, by means of its Fourier series: $\phi(z)=\sum_{n \in \Z}\widehat{\phi}(n)z^n$, with Fourier coefficients $\widehat{\phi}(n)= \int_{\mathbb{T}}\phi(z) \bar{z}^n dz, \, n \in \Z$. This establishes an isomorphism between the function $\phi(z)$ and the sequence of the corresponding Fourier coefficients $\widehat{\phi}$. Thus, we can partition the function space $\mathcal{L}^2(\mathbb{T})$ into two subspaces.
    
    \begin{definition}
        For $0<p\leq\infty$ , the \textbf{Hardy space} $\mathcal{H}^p$ on $\mathbb{T}$ is the subspace of $\mathcal{L}^p(\mathbb{T})$ defined as:
        \begin{equation}
            \mathcal{H}^p= \{ \phi(z) \in \mathcal{L}^p(\mathbb{T}) : \widehat{\phi}(n)=0, n < 0\},
        \end{equation}
        while the \textbf{negative Hardy space} on $\mathbb{T}$ is the subspace of $\mathcal{L}^p(\mathbb{T})$
        \begin{equation}
            \mathcal{H}^p_-=\{ \phi(z) \in  \mathcal{L}^p(\mathbb{T}) : \widehat{\phi}(n)=0, n \geq 0\}.
        \end{equation}
    \end{definition} 
    
    It is possible to define Hardy spaces also on the open unit disc $\mathbb{D}$.

    \begin{definition}
        The \textbf{Hardy space} $\mathcal{H}^p(\mathbb{D})$ on $\mathbb{D}$ for $0<p<\infty$ consists of functions $\phi(z)$ analytic in $\mathbb{D}$ and such that:
        \begin{equation}\label{eq:hardynorm}
            \norm{\phi}_p:=\sup_{0<r<1}\Bigg(\int_{\mathbb{T}}|\phi(r\xi)|^pdm(\xi)\Bigg)^{1/p} < \infty
        \end{equation}
        and it is equipped with the norm $\norm{\cdot}_p$. For $p=\infty$, $\mathcal{H}^{\infty}(\mathbb{D})$ is the space of bounded analytic functions in $\mathbb{D}$ with norm:
        \begin{equation}
            \norm{\phi}_{\infty}:=\sup_{\xi \in \mathbb{D}}|\phi(\xi)|.
        \end{equation}
    \end{definition}
    
    Interestingly, $\mathcal{H}^p(\mathbb{D})$ and $\mathcal{H}^p$ can be canonically identified  by associating a function $\phi(z) \in\mathcal{H}^p(\mathbb{D})$ with its limit on the boundary, which is a function in $\mathcal{H}^p$ (a proof can be found in~\cite{Nikolski}). Thus, we will make no difference between those functions in the unit disc and their boundary value on the circle.
    
    We can now embed the sequence space $\ell^2$ into $\ell^2(\Z)$ by ``duplicating'' each vector, \emph{i.e.} by associating $\boldsymbol{\mu}=(\mu_0, \mu_1, \dots)\in\ell^2$ to $\boldsymbol{\mu}^{(2)}=(\dots, \mu_1,\mu_0, \mu_1, \dots)\in\ell^2(\Z)$. Then, we can use the Fourier isomorphism to map the vector $\boldsymbol{\mu}^{(2)}\in\ell^2(\Z)$ to the function space $\mathcal{L}^2(\mathbb{T})$. In this way each vector $\boldsymbol{\mu}\in \ell^2$ corresponds to two functions in the Hardy spaces:  
    \begin{align}\label{eq:hardynotation}
        &\mu^-(z)=\sum_{j=0}^{\infty}\boldsymbol{\mu}_j z^{-j-1} \in \mathcal{H}^2_-,\\
        &\mu^+(z)=\sum_{j=0}^{\infty} \boldsymbol{\mu}_j z^{j} \in \mathcal{H}^2.
    \end{align}

    This leads to an alternative characterization of Hankel operators in Hardy spaces. 
    \begin{definition}\label{Hankel2}
        Let $\phi(z)$ be a function in the space $ \mathcal{L}^2(\mathbb{T})$. A \textbf{Hankel operator} is an operator $H_{\phi}:\mathcal{H}^2 \rightarrow \mathcal{H}^2_-$ defined by $ H_{\phi}f(z)=\mathbb{P}_-\phi f(z)$, where $\mathbb{P}_-$ is the orthogonal projection from $ \mathcal{L}^2(\mathbb{T})$ onto $\mathcal{H}^2_- $ . The function $\phi(z)$ is called a \textbf{symbol} of the Hankel operator $H_{\phi}$.
    \end{definition}
    
    
    If $H_{\phi}$ is a bounded operator, we can consider without loss of generality $\phi(z) \in \mathcal{L}^{\infty}(\mathbb{T})$. This is a consequence of Nehari's theorem~\cite{Nehari}, whose formulation can be found in Appendix~\ref{appendix:fun_an}, together with more details about the two definitions of Hankel operators. We remark that a Hankel operator has infinitely many different symbols, since $H_{\phi}=H_{\phi+\psi}$ for $\psi(z) \in \mathcal{H}^{\infty}$.
    
    \begin{remark}
    In the standard orthonormal bases, $\{z^k\}_{k\geq 0}$ in $\mathcal{H}^2$ and $\{{z}^{-(j+1)}\}_{j\geq 0}$ in $\mathcal{H}^2_-$, the Hankel operator $H_{\phi}$ has Hankel matrix $\H(j,k)= \widehat{\phi}(-j-k-1)$ for $j,k\geq0$. 
    \end{remark}
    
    \begin{example}\label{example2}
        In the case of the Hankel matrix in Example~\ref{example1}, since  $\H(j,k)= \widehat{\phi}(-j-k-1)$, we have:
        \begin{equation*}
             \H=
            \begin{pmatrix} \frac{3}{4} & 0 & \frac{3}{16} & \dots \\
                               0 & \frac{3}{16} & 0 &\dots \\
                               \frac{3}{16} & 0 & \frac{3}{64} &\dots \\
                                \vdots& \vdots &\vdots&\ddots
            \end{pmatrix}
            =\begin{pmatrix} \widehat{\phi}(-1) & \widehat{\phi}(-2) & \widehat{\phi}(-3) & \dots \\
                               \widehat{\phi}(-2) & \widehat{\phi}(-3) & \widehat{\phi}(-4) &\dots \\
                               \widehat{\phi}(-3) & \widehat{\phi}(-4) & \widehat{\phi}(-5) &\dots \\
                                \vdots& \vdots &\vdots&\ddots
            \end{pmatrix}.
        \end{equation*} 
        Hence, we can recover the corresponding symbol: 
        \begin{equation*}
        \mathbb{P}_-\phi=\sum_{n \geq 0}\widehat{\phi}(-n-1)z^{-2n-1}=\sum_{n \geq 0}\frac{3}{4}4^{-n}z^{-n-1}=\frac{3z}{4z^2-1}.
        \end{equation*}
    \end{example}
    
    \begin{definition}\label{def:rational}
        The complex function $f(z)$ is \textbf{rational} if $f(z)=p(z)/q(z)$, with $p(z)$ and $q(z)$ polynomials. The rank of $f(z)$ is the maximum between the degrees of $p(z)$ and $q(z)$. A rational function is \textbf{strictly proper} if the degree of $p(z)$ is strictly smaller than that of $q(z)$.
    \end{definition}
    We remark that the poles of a complex function $f$ correspond to the zeros of $1/f$. The following result of Kronecker relates finite-rank infinite Hankel matrices to rational functions.

    \begin{theorem}[\cite{kronecker}]\label{theorem:Kronecker}
        Let $H_{\phi}$ be a bounded Hankel operator with matrix $\H$. Then $\H$ has finite rank if and only if $\mathbb{P}_-\phi$ is a strictly proper rational function. Moreover the rank of $\H$ is equal to the number of poles (with multiplicities) of $\mathbb{P}_-\phi$ inside the unit disc.
    \end{theorem}
    
    \begin{example}
        The function in Example~\ref{example2} is rational with degree $2$ and has two poles inside the unit disc at $z=\pm\frac{1}{2}$. Thus, the Hankel matrix associated has rank 2.
    \end{example}
    
    We state as remark an important takeaway from this section.
    
    \begin{remark}\label{remark}
        Given a rank $n$ Hankel matrix $\H$, we can look at it in two alternative ways. On the one hand we can consider the Hankel operator over sequences $H_f:\ell^2\rightarrow\ell^2$, associated to a function $f:\Sigma^* \rightarrow \R$. In this case $\H(i,j)=f(i+j)$ for $i,j\geq0$, and $f$ is rational in the sense that it is realized by a WFA of size $n$. On the other hand, we can consider the Hankel operator over complex (Hardy) spaces $H_{\phi}:\mathcal{H}^2 \rightarrow \mathcal{H}^2_-$, associated to a function $\phi(z)\in  \mathcal{L}^2(\mathbb{T})$, the symbol. In this case $\H(j,k)= \widehat{\phi}(-j-k-1)$ for $j,k\geq0$, and $\mathbb{P}_-\phi=\widehat{\phi}(-j-k-1)$ is rational of rank $n$ in the sense of Definition~\ref{def:rational}.
    \end{remark}


    We can introduce now the main result of Adamyan, Arov and Krein~\cite{AAK71}. The theorem, stated for Hankel operators over Hardy spaces, shows that for infinite dimensional Hankel matrices the constraint of preserving the Hankel property does not affect the achievable approximation error.
    \begin{theorem}[AAK-1\cite{AAK71}]\label{theorem:aakop}
        Let $H_{\phi}$ be a compact Hankel operator of rank $n$ and singular numbers $\sigma_m$, with $0 \leq m < n$ and $\sigma_0 \geq \dots \geq \sigma_{n-1}>0$. Then there exists a unique Hankel operator $G$ of rank $k<n$ such that:
        \begin{equation}\label{eqoper}
            \norm{\H - \mat{G}}= \sigma_k.
        \end{equation}
    \end{theorem}
    

    We denote with $\mathcal{R}_k\subset \mathcal{H}^{\infty}_-$ the set of strictly proper rational functions of rank $k$, and we consider the set:
    \begin{equation}
        \mathcal{H}^{\infty}_k=\{\psi(z) \in \mathcal{L}^{\infty}(\mathbb{T}): \,\,\exists g(z) \in \mathcal{R}_k, \exists l(z) \in \mathcal{H}^{\infty},\,\, \psi(z)=g(z)+l(z) \}.
    \end{equation}
    We can reformulate the theorem in terms of symbols. 
    \begin{theorem}[AAK-2~\cite{AAK71}]\label{theorem:aaksymb}
        Let $\phi(z) \in \mathcal{L}^{\infty}(\mathbb{T})$. Then there exists $\psi(z) \in \mathcal{H}^{\infty}_k$ such that:
        \begin{equation}\label{eqsymbol}
            \norm{\phi(z) - \psi(z)}_{\infty}= \sigma_k(H_{\phi}).
        \end{equation}
    \end{theorem}

    The solutions of Theorem~\ref{theorem:aakop} and~\ref{theorem:aaksymb} are strictly related (proof in Appendix~\ref{appendix:fun_an}).

    \begin{corollary}\label{corollary:stable}
        Let $\psi(z)=g(z)+l(z) \in \mathcal{H}^{\infty}_k$, with $g(z) \in \mathcal{R}_k, \, l(z) \in \mathcal{H}^{\infty}$. If $\psi(z)$ solves Equation~\ref{eqsymbol}, then $G=H_g$ is the unique Hankel operator from Theorem~\ref{theorem:aakop}.
    \end{corollary}
    
    We state as corollary the key point of the proof of AAK Theorem, that provides us with a practical way to find the best approximating symbol.
    \begin{corollary}\label{corollary:unimodular}
        Let $\phi(z)$ and $\{\boldsymbol{\xi}_k, \boldsymbol{\eta}_k\}$ be a symbol and a $\sigma_k$-Schmidt pair for $H_{\phi}$. A function $\psi(z) \in \mathcal{L}^{\infty}$ is the best AAK approximation according to Theorem~\ref{theorem:aaksymb}, if and only if:
        \begin{equation}\label{eq:unimodular}
            (\phi(z)-\psi(z))\xi^{+}_k(z)=\sigma_k\eta^{-}_k(z).
        \end{equation}
        Moreover, the function $\psi(z)$ does not depend on the particular choice of the pair $\{\boldsymbol{\xi}_k, \boldsymbol{\eta}_k\}$.
    \end{corollary}


\section{Approximate Minimization}\label{section3}

\subsection{Assumptions}
    
    To apply AAK theory to the approximate minimization problem, we consider only automata with real weights, defined over a one-letter alphabet. In this case, the free monoid generated by the single letter can be identified with $\N$, and canonically embedded into $\Z$. This passage is fundamental to use Fourier analysis and the isomorphism that leads to Theorem~\ref{theorem:aakop} and \ref{theorem:aaksymb}. Unfortunately, this idea cannot be directly generalized to bigger alphabets, since in this case we would obtain a free non-abelian monoid (not identifiable with $\Z$).
    
    Theorem~\ref{theorem:aakop} requires the Hankel operator $H$ to be bounded. To ensure that a minimal WFA $A= \langle \balpha, \A, \bbeta\rangle$ satisfies this condition, we assume $\rho(\A)<1$, where $\rho$ is the spectral radius of $\A$~\cite{Balle19}. As a matter of fact, to guarantee the boundness of the Hankel operator it is enough that the considered WFA computes a function $f\in\ell^2$~\cite{Balle19}. However, the stricter assumption on the spectral radius is needed when computing the symbol associated to a WFA. This condition directly implies the existence of the SVA, and of the Gramian matrices $\mP$ and $\mQ$, with $\mP=\mQ$ diagonal~\cite{Balle19}. We assume that $A= \langle \balpha, \A, \bbeta\rangle$ is in SVA form. In this case, given the size of the alphabet, the Hankel matrix $\H$ is symmetric, so $\A=\A^{\top}$. Moreover, if we denote with $\lambda_i$ the $i$-th non-zero eigenvalue of $\H$, and we consider the coordinates of $\balpha$ and $\bbeta$, we have that $\balpha_i=\operatorname{sgn}(\lambda_i)\bbeta_i$, where $\operatorname{sgn}(\lambda_i)=\lambda_i/|\lambda_i|$.
    
    For example, we note that a minimal GPA computes a function $f\in \ell^1$, so the condition on $\rho(\A)$ is automatically satisfied by this class of WFAs~\cite{Balle19}. Possible relaxations of the spectral radius assumption are discussed in Appendix~\ref{appendix:relax}, together with an alternative method to find the optimal spectral-norm approximation of a Hankel matrix without extracting a WFA.
    
    Finally, in this paper we only consider automata with weights in $\R$, though results remain true for complex numbers. The method we present can be easily extended to vector-valued automata~\cite{rabusseaumultitask}, but the solution to the optimal approximation problem will not be unique~\cite{Peller}.

\subsection{Problem Formulation}

    Let $A= \langle \balpha, \A, \bbeta \rangle$ be a minimal WFA with $n$ states in SVA form, defined over a one-letter alphabet. Let $\H$ be the Hankel matrix of $A$, we denote with $\sigma_i$, for $0\leq i < n$, the singular numbers. Given a target number of states $k<n$, we say that a WFA $\widehat{A}_k$ with $k$ states solves the \emph{optimal spectral-norm approximate minimization} problem if the Hankel matrix $\mat{G}$ of $\widehat{A}_k$ satisfies $\norm{\H - \mat{G}}= \sigma_k(\H)$. Note that the content of the ``optimal spectral-norm approximate minimization'' is equivalent to the problem solved by Theorem~\ref{theorem:aakop}, with the exception that here we insist on representing the inputs and outputs of the problem effectively by means of WFAs. Based on the AAK theory sketched in Section~\ref{aak-section}, we draw the following steps:
    
    \begin{enumerate}
        \item \emph{Compute a symbol $\phi(z)$ for $H$ using Remark~\ref{remark}}. We obtain the negative Fourier coefficients of $\phi(z)$ and derive its Fourier series.
        \item\label{step2} \emph{Compute the optimal symbol $\psi(z)$ using Corollary~\ref{corollary:unimodular}}. The main challenge here is to find a suitable representation for the functions $\psi(z)$ and $e(z)=\phi(z)-\psi(z)$. We define them in terms of two auxiliary WFAs. The key point is to select constraints on their parameters to leverage the properties of weighted automata, while still keeping the formulation general.  
        \item \emph{Extracting the rational component by solving for $g(z)$ in Corollary~\ref{corollary:stable}}. This step is arguably the most conceptually challenging, as it requires to identify the position of the function's poles. In fact, we know from Theorem~\ref{theorem:Kronecker} that $g(z)$ has $k$ poles, all inside the unit disc.
        \item \emph{Find a WFA representation for $g(z)$}. Since in Step~\ref{step2} we parametrized the functions using WFAs, the expression of $g(z)$ directly reveals the WFA $\widehat{A}_k$.
    \end{enumerate}

\subsection{Spectral-Norm Approximate Minimization}\label{method}

    In the following sections we will consider a minimal WFA $A= \langle \balpha, \A, \bbeta \rangle$ with $n$ states in SVA form, defined over a one-letter alphabet $\Sigma=\{a\}$, its Hankel matrix $\H$, corresponding to the bounded operator $H$, and the singular numbers $\sigma_i$, for $0\leq i < n$. Let $f: \Sigma^* \rightarrow \R $ be the function realized by $A$. We denote by $x$ the string where $a$ is repeated $x$ times, so we have $f(x)=\balpha^{\top} \A^x \bbeta$.

\subsubsection{Computation of a Symbol for A}  

\label{sec:symbol.computation}

    To determine the symbol $\phi(z)$ of $H$, we recall that each entry of the Hankel matrix corresponds simultaneously to the values of $f$ and to the negative Fourier coefficients of $\phi(z)$. In fact, as seen in Remark~\ref{remark}, we have:
    \begin{equation}
        \H= \begin{pmatrix} f_A(0) & f_A(1) & f_A(2) & \dots \\
                               f_A(1) & f_A(2) & f_A(3) &\dots \\
                               f_A(2) & f_A(3) & f_A(4) &\dots \\
                                \vdots& \vdots &\vdots&\ddots
            \end{pmatrix}
            =\begin{pmatrix} \widehat{\phi}(-1) & \widehat{\phi}(-2) & \widehat{\phi}(-3) & \dots \\
                               \widehat{\phi}(-2) & \widehat{\phi}(-3) & \widehat{\phi}(-4) &\dots \\
                               \widehat{\phi}(-3) & \widehat{\phi}(-4) & \widehat{\phi}(-5) &\dots \\
                                \vdots& \vdots &\vdots&\ddots
            \end{pmatrix}.
    \end{equation} 
    We obtain:
    \begin{equation}\label{eq:symbolsum}
        \mathbb{P}_-\phi(z)= \sum_{k\geq 0}f(k) z^{-k-1} = \sum_{k\geq 0} \balpha^{\top}\A^k \bbeta z^{-k-1} = \balpha^{\top}(z\mat{1}-\A)^{-1} \bbeta
    \end{equation}
    where we use the fact that $\rho(A)<1$ for the last equality. Since the function obtained is already bounded, we can directly consider $\phi(z)=\balpha^{\top}(z\mat{1}-\A)^{-1} \bbeta$ as a symbol for $H$.
    
    \begin{example}
        If we apply the formula in Equation~\ref{eq:symbolsum} to the GPA in Example~\ref{example1}, we recover the rational function $\phi(z)= \frac{3z}{4z^2-1}$ found in Example~\ref{example2}.
    \end{example}

\subsubsection{Computation of the Optimal Symbol}    
    
    We consider two auxiliary WFAs. Let $\widehat{A}=\langle\widehat{\balpha},\widehat{\A},\widehat{\bbeta}\rangle$ be a WFA with $j\geq k$ states, satisfying the following properties:
    \begin{enumerate}
        \item $1$ is not an eigenvalue of $\widehat{\A}$
        \item\label{Eminimal} the automaton  $E= \langle \balpha_e, \A_e, \bbeta_e \rangle$ is minimal, with
        \begin{equation}\label{eq:partition.Ae}
            \A_e= \begin{pmatrix}
                    \A & \mat{0}    \\
                    \mat{0}  & \widehat{\A}  
                \end{pmatrix} , \quad
            \balpha_e= \begin{pmatrix}
                \balpha          \\
                -\widehat{\balpha}  
                \end{pmatrix}, \quad
            \bbeta_e= \begin{pmatrix}
                \bbeta          \\
                \widehat{\bbeta}  
                \end{pmatrix}.
        \end{equation}
    \end{enumerate}
    
    Using the parameters of the automaton $\widehat{A}$ and a constant $C$, we define a function $\psi(z)= \widehat{\balpha}^{\top}(z\mat{1}-\widehat{\A})^{-1} \widehat{\bbeta}+C$. We remark that the poles of $\psi(z)$ correspond to the eigenvalues of $\widehat{\A}$, counted with their multiplicities. By assumption, $1$ is not an eigenvalue of $\widehat{A}$, so $\psi(z)$ does not have any poles on the unit circle, and therefore $\psi(z) \in \mathcal{L}^{\infty}(\mathbb{T})$. Analogously, the function $e(z)= \phi(z)-\psi(z)= \balpha_e^{\top}(z\mat{1}-\A_e)^{-1} \bbeta_e - C$ is also bounded on the circle. 

    Our objective is to compute the parameters of $\widehat{A}=\langle\widehat{\balpha},\widehat{\A},\widehat{\bbeta}\rangle$ that make $\psi(z)$ the best approximation of $\phi(z)$ according to Theorem~\ref{theorem:aaksymb}. In particular, we will use Corollary~\ref{corollary:unimodular} to find the triple $\widehat{\balpha},\widehat{\A},\widehat{\bbeta}$ such that $\psi(z)$ satisfies Equation~\ref{eq:unimodular}. Note that, with this purpose, the constant term $C \in H^{\infty}$ becomes necessary to characterize $\psi(z)$. In fact, while the {$H^{\infty}$-component} of the symbol does not affect the Hankel norm, it plays a role in the computation of the $\mathcal{L}^{\infty}$-norm (in Equation~\ref{eqsymbol}) according to Nehari (Theorem~\ref{thm:nehari1}), so it cannot be dismissed.


    The following theorem provides us with an explicit expression for the functions in the Hardy space corresponding to a $\sigma_k$-Schmidt pair.
    
    \begin{theorem}\label{theorem:singval}
        Let $\sigma_k$ be a singular number of the Hankel operator $H$. The singular functions associated with the $\sigma_k$-Schmidt pair $\{\boldsymbol{\xi}_k, \boldsymbol{\eta}_k\}$ of $H$ are:
        \begin{align}
            \xi^{+}_k(z)&=\sigma_k^{-1/2}\balpha^{\top}(\mat{1} - z\A)^{-1}\mat{e}_k\\
            \eta^{-}_k(z)&=\sigma_k^{-1/2}\bbeta^{\top}(z\mat{1}-\A)^{-1}\mat{e}_k.
        \end{align}
        If $\psi(z)$ is the best approximation to the symbol, then $\sigma_k^{-1}e(z)$ has modulus $1$ almost everywhere on the unit circle (\emph{i.e.} it is unimodular).
    \end{theorem}
    
    \begin{proof}
    
        Let $\mat{F}$ and $\mat{B}$ be the forward and backward matrices, respectively, with $\H=\mat{F}\mat{B}^{\top}$, $\mP=\mat{F}^\top \mat{F}, \mQ=\mat{B}^\top \mat{B}$. We consider the $\sigma_k$-Schmidt pair $\{\boldsymbol{\xi}_k, \boldsymbol{\eta}_k\}$. By definition, $\H^{\top}\H\boldsymbol{\xi}_k=\sigma_k^2\boldsymbol{\xi}_k$. Rewriting in terms of the FB factorization, we obtain:
        \begin{align}
            &\H^{\top}\H\boldsymbol{\xi}_k=\sigma_k^2\boldsymbol{\xi}_k\\
            &\mat{B}\mat{F}^{\top}\mat{F}\mat{B}^{\top}\boldsymbol{\xi}_k=\sigma_k^2\boldsymbol{\xi}_k\\
            &\mat{B}\mP\mat{B}^{\top}\boldsymbol{\xi}_k=\sigma_k^2\boldsymbol{\xi}_k\\
            &\mat{B}\mP\mat{e}_k=\sigma_k^2\boldsymbol{\xi}_k
        \end{align} 
        where in the last step we set $\mat{e}_k= \mat{B}^ \top \boldsymbol{\xi}_k$, to reduce the SVD problem of $\H$ to the one of $\mQ\mP$. Note that, since $\mP$ and $\mQ$ are diagonal, $\mat{e}_k$ is the $k$-th coordinate vector $(0,\dots,0,1,0,\dots,0)^{\top}$. Since $\mat{e}_k$ is an eigenvector of $\mQ\mP$ for $\sigma_k^2$, we get:
        \begin{align}
            &\mat{B}\mQ^{-1}\mQ\mP\mat{e}_k=\sigma_k^2\boldsymbol{\xi}_k\\
            &\mat{B} \mQ^{-1} \mat{e}_k= \boldsymbol{\xi}_k.
        \end{align}
        Moreover, $\H$ is symmetric, so we have that the singular vectors $\boldsymbol{\eta}_k$ and $\boldsymbol{\xi}_k$ have the same coordinates up to the sign of the corresponding eigenvalues.
        We obtain:
        \begin{align}
            \xi^+_k(z)&=\sum_{j=0}^{\infty} \sigma_k^{-1/2}\balpha^{\top}\A^j \mat{e}_k z^{j}= \sigma_k^{-1/2}\balpha^{\top}(\mat{1} - z\A)^{-1}\boldsymbol{e}_k\\
            \eta^{-}_k(z)&=\sum_{j=0}^{\infty}\sigma_k^{-1/2}\bbeta^{\top}\A^j \mat{e}_k z^{-j-1}=\sigma_k^{-1/2}\bbeta^{\top}(z\mat{1}-\A)^{-1}\boldsymbol{e}_k
        \end{align}
        where the singular functions have been computed following Equation~\ref{eq:hardynotation}. If $r$ is the multiplicity of $\sigma_k$, from Corollary~\ref{corollary:unimodular} we get the following fundamental equation:
        \begin{equation*}
            (\phi(z)-\psi(z))\balpha^{\top}(\mat{1} - z\A)^{-1}\mat{V}= \sigma_k \bbeta^{\top}(z\mat{1}-\A)^{-1}\mat{V}
        \end{equation*}
        where $\mat{V}=\begin{pmatrix}
            \mat{0} & \mat{1}_r
          \end{pmatrix}^\top$ is a $n\times r$ matrix.
        Consequently, we obtain the function:
        \begin{equation*}
            \sigma_k^{-1}e(z)= \frac{\bbeta^{\top}(z\mat{1}-\A)^{-1}\mat{V}}{\balpha^{\top}(\mat{1} - z\A)^{-1}\mat{V}}
        \end{equation*}
        which is unimodular, since $\balpha_i=\operatorname{sgn}(\lambda_i)\bbeta_i$.
    \end{proof} 
    
    It is reasonable to wonder how the fact that $\sigma_k^{-1}e(z)$ is unimodular reflects on the structure of the WFA $E=\langle \balpha_e, \A_e, \bbeta_e\rangle$ associated with it. We remark that, \emph{a priori}, the controllability and observability Gramians of $E$ might not be well defined. The following theorem provides us with two matrices $\mP_e$ and $\mQ_e$ satisfying properties similar to those of the Gramians. This theorem is the analogous of a control theory result~\cite{discreteH}, rephrased in terms of WFAs. A sketch of the proof, that relies on the minimality of the WFA $E$~\cite{schutter}, can be found in Appendix~\ref{appendix:proofs}. For the detailed version of the proof and the original theorem we refer the reader to~\cite{discreteH}.
    
    \begin{theorem}[\cite{discreteH}]\label{Theorem:allpass}
        Consider the function $e(z)= \balpha_e^{\top}(z\mat{1}-\A_e)^{-1} \bbeta_e - C$ and the corresponding minimal WFA $E=\langle \balpha_e, \A_e, \bbeta_e \rangle$ associated with it. If $\sigma_k^{-1}e(z)$ is unimodular, then there exists a unique pair of symmetric invertible matrices $\mP_e$ and $\mQ_e$ satisfying:
        \begin{itemize}
            \item[(a)] $\mP_e-\A_e \mP_e \A_e^\top = \bbeta_e\bbeta_e^\top$ \label{a}
            \item[(b)] $\mQ_e-\A_e^\top \mQ_e \A_e = \balpha_e\balpha_e^\top$ \label{b}
            \item[(c)] $\mP_e\mQ_e=\sigma^2_k\mat{1}$ \label{c}
        \end{itemize}
    \end{theorem}

    We can now derive the parameters of the WFA $\widehat{A}=\langle\widehat{\balpha},\widehat{\A},\widehat{\bbeta}\rangle$.
    
    \begin{theorem}\label{Theorem:maintrm}
    Let $A=\langle \balpha, \A, \bbeta \rangle $ be a minimal WFA with $n$ states in its SVA form, and let $\phi(z)= \balpha^{\top}(z\mat{1}-\A)^{-1}\bbeta$ be a symbol for its Hankel operator $H$. Let $\sigma_k$ be a singular number of multiplicity $r$ for $H$, with $\sigma_0 \geq \dots > \sigma_k=\dots=\sigma_{k+r-1}>\sigma_{k+r}\geq \dots \geq \sigma_{n-1}>0.$
    We can partition the Gramian matrices $\mP$, $\mQ$ as follows:
    \begin{equation}\label{eq:setup}
        \mP=\mQ=\begin{pmatrix}
            \mat{\Sigma} & \mat{0}     \\
            \mat{0}      & \sigma_k\mat{1}_r
            \end{pmatrix},
    \end{equation}
    where $\boldsymbol{\Sigma}\in \R^{(n-r)\times(n-r)}$ is the diagonal matrix containing the remaining singular numbers, and partition $\A$, $\balpha$ and $\bbeta$ conformally to the Gramians:
    \begin{equation}\A= \begin{pmatrix}
            \A_{11}          & \A_{12} \\
            \A_{12}^{\top}   & \A_{22} 
            \end{pmatrix},\quad
        \balpha= \begin{pmatrix}
             \balpha_1 \\
             \balpha_2 
            \end{pmatrix},\quad
        \bbeta= \begin{pmatrix}
             \bbeta_1 \\
             \bbeta_2 
            \end{pmatrix}.
    \end{equation}
    Let $\mat{R}=\sigma_k^2\mat{1}_{n-r}-\mat{\Sigma}^2$, we denote by $(\cdot)^{+}$ the Moore-Penrose pseudo-inverse. If the function $\psi(z)= \widehat{\balpha}^{\top}(z\mat{1}-\widehat{\A})^{-1} \widehat{\bbeta}+C$ is the best approximation of $\phi(z)$, then:
    \begin{itemize}
        \item  If $\balpha_2 \neq \mat{0}$:
        \begin{equation}
            \begin{cases}
                \widehat{\bbeta} = - \widehat{\A}\A_{12}(\bbeta_2^{\top})^{+} \\
                \widehat{\balpha} = \widehat{\A}^{\top}\mat{R}\A_{12}(\balpha_2^{\top})^{+}\\
                \widehat{\A}(\A_{11}- \A_{12}(\bbeta_2^{\top})^{+}\bbeta_1^{\top})=\mat{1}
        \end{cases}
        \end{equation}
        \item If $\balpha_2=\mat{0}$:
        \begin{equation}\label{eq:balphazero}
            \begin{cases}
                \widehat{\bbeta} =  (\mat{1} -\widehat{\A}\A_{11})(\bbeta_1^{\top})^{+}\\
                \widehat{\balpha} =-(\mat{R} -\widehat{\A}^{\top}\mat{R}\A_{11})(\balpha_1^{\top})^{+}\\
                \widehat{\A}\A_{12}=\mat{0}
            \end{cases}
        \end{equation}
    \end{itemize}
    \end{theorem}
    
    \begin{proof}[Proof of Theorem~\ref{Theorem:maintrm}]
    Since $\sigma^{-1}e(z)=\phi(z)-\psi(z)$ is unimodular, from Theorem~\ref{Theorem:allpass} there exist two symmetric nonsingular matrices $\mP_e$, $\mQ_e$ satisfying the fixed point equations:
        \begin{align}\label{fixed_point_e}
            \mP_e-\A_e\mP_e\A_e^{\top}&=\bbeta_e\bbeta_e^{\top}\\
            \mQ_e-\A_e^{\top}\mQ_e\A_e&=\balpha_e\balpha_e^{\top}\label{fixed_point_e2}
        \end{align}
    and such that $\mP_e\mQ_e=\sigma^2_k\mat{1}$.
    We can partition $\mP_e$ and $\mQ_e$ according to the definition of $\A_e$~(see Eq.~\ref{eq:partition.Ae}):
        \begin{equation*}
            \mP_e= \begin{pmatrix}
                \mP_{11}          & \mP_{12} \\
                \mP_{12}^{\top}   & \mP_{22} 
                \end{pmatrix},\quad
            \mQ_e= \begin{pmatrix}
                 \mQ_{11}          & \mQ_{12} \\
                \mQ_{12}^{\top}   & \mQ_{22} 
          \end{pmatrix}.
        \end{equation*}
    From Equation~\ref{fixed_point_e} and~\ref{fixed_point_e2}, we note that $\mP_{11}$ and $\mQ_{11}$ corresponds to the controllability and observability Gramians of $A$: 
        \begin{equation*}
            \mP_{11}=\mQ_{11}=\mP=\begin{pmatrix}
                          \mat{\Sigma} & \mat{0}     \\
                          \mat{0}      & \sigma_k\mat{1}
                          \end{pmatrix}.
        \end{equation*}
    Moreover, since $\mP_e\mQ_e=\sigma_k^2\mat{1}$, we get $\mP_{12}\mQ_{12}^{\top}=\sigma_k^2\mat{1}-\mP^2$. It follows that $\mP_{12}\mQ_{12}^{\top}$ has rank $n-r$. Without loss of generality we can set $\dim{\widehat{\A}}=j=n-r$, and choose an appropriate basis for the state space such that $\mP_{12}=\begin{pmatrix} \mat{1} &  \mat{0} \end{pmatrix} ^{\top}$ and $\mQ_{12}=\begin{pmatrix} \mat{R} &  \mat{0} \end{pmatrix} ^{\top}$, with $\mat{R}=\sigma_k^2\mat{1}-\mat{\Sigma}^2$. Once $\mP_{12}$ and $\mQ_{12}$ are fixed, the values of $\mP_{22}$ and $\mQ_{22}$ are automatically determined. We obtain:
    \begin{equation}
        \mP_e= \begin{pmatrix}
            \mat{\Sigma}     & \mat{0}                    & \mat{1}\\
            \mat{0}          & \sigma_k\mat{1}     & \mat{0}\\
             \mat{1} & \mat{0}                    & -\mat{\Sigma} \mat{R}^{-1}  
             \end{pmatrix},\quad
        \mQ_e= \begin{pmatrix}
            \mat{\Sigma}     & \mat{0}                    & \mat{R}\\
            \mat{0}          & \sigma_k\mat{1}     & \mat{0}\\
            \mat{R}         & \mat{0}                    & -\mat{\Sigma} \mat{R}
            \end{pmatrix}.
    \end{equation}
    
    Now that we have an expression for the matrices $\mP_e$ and $\mQ_e$ of Theorem~\ref{Theorem:allpass}, we can  rewrite the fixed point equations to derive the parameters $\widehat{\balpha}$, $\widehat{\A}$ and $\widehat{\bbeta}$. We obtain the following systems:
        \begin{equation}
            \begin{cases}
                \mP -\A \mP \A= \bbeta\bbeta^{\top}\\
                \mat{N} -\A \mat{N} \widehat{\A}^{\top}= \bbeta\widehat{\bbeta}^{\top}\\
                -\mat{\Sigma} \mat{R}^{-1} +\widehat{\A}\mat{\Sigma} \mat{R}^{-1}\widehat{\A}^{\top}={\widehat{\bbeta}}{\widehat{\bbeta}}^{\top}
             \end{cases} \quad  \begin{cases}
                \mP -\A \mP \A= \balpha\balpha^{\top}\\
                \mat{M} -\A^{\top} \mat{M} \widehat{\A}= -\balpha\widehat{\balpha}^{\top}\\
                -\mat{\Sigma} \mat{R} +\widehat{\A}^{\top}\mat{\Sigma} \mat{R}\widehat{\A}={\widehat{\balpha}}{\widehat{\balpha}}^{\top}
            \end{cases}
        \end{equation}
    where $\mat{N}= \begin{pmatrix}
                \mat{1}     \\
                \mat{0}
            \end{pmatrix}$ and $\mat{M}= \begin{pmatrix}
                                         \mat{R}     \\
                                         \mat{0}
                                          \end{pmatrix}$. We can rewrite the second equation of each system as follows:
        \begin{equation}
            \begin{cases}
                \mat{1} -\A_{11} \widehat{\A}^{\top}= \bbeta_1\widehat{\bbeta}^{\top}\\
                -\A_{12}^{\top}\widehat{\A}^{\top}=\bbeta_2\widehat{\bbeta}^{\top}
             \end{cases} \quad  \begin{cases}
                \mat{R} -\A_{11} \mat{R} \widehat{\A}=-\balpha_1\widehat{\balpha}^{\top}\\
                \widehat{\A}^{\top}\mat{R}\A_{12}=\widehat{\balpha}\balpha_2^{\top}
            \end{cases}
        \end{equation}        
    If $\balpha_2 \neq \mat{0}$, then also $\bbeta_2 \neq \mat{0}$ (recall that $\balpha_i=\operatorname{sgn}(\lambda_i)\bbeta_i$), and we have:                        
    \begin{equation}
        \begin{cases}
                \widehat{\bbeta} = - \widehat{\A}\A_{12}(\bbeta_2^{\top})^{+} \\
                \widehat{\balpha} = \widehat{\A}^{\top}\mat{R}\A_{12}(\balpha_2^{\top})^{+}\\
                \widehat{\A}(\A_{11}- \A_{12}(\bbeta_2^{\top})^{+}\bbeta_1^{\top})=\mat{1}
        \end{cases}
    \end{equation}
    with $(\balpha_2^{\top})^{+}=\frac{\balpha_2}{\balpha_2^{\top}\balpha_2}$ and $(\bbeta_2^{\top})^{+}=\frac{\bbeta_2}{\bbeta_2^{\top}\bbeta_2}$.
    
    If $\balpha_2=\mat{0}$, we have $\widehat{\A}\A_{12}=\mat{0}$. We remark that $\widehat{\A}$ has size $(n-r)\times(n-r)$, while $\A_{12}$ is $(n-r)\times r$, so the system of equations corresponding to $\widehat{\A}\A_{12}=\mat{0}$ is underdetermined if $r<\frac{n}{2}$, in which case we can find an alternative set of solutions:
    \begin{equation}\label{eq:case2}
            \begin{cases}
                \widehat{\bbeta} =  (\mat{1} -\widehat{\A}\A_{11})(\bbeta_1^{\top})^{+}\\
                \widehat{\balpha} =-(\mat{R} -\widehat{\A}^{\top}\mat{R}\A_{11})(\balpha_1^{\top})^{+}\\
                \widehat{\A}\A_{12}=\mat{0}
            \end{cases}
    \end{equation}
    with $\widehat{\A}\neq \mat{0}$. On the other hand, if $r\geq\frac{n}{2}$, \emph{i.e.} if the multiplicity of the singular number $\sigma_k$ is more than half the size of the original WFA, the system might not have any solution unless $\widehat{\A}=\mat{0}$ (or unless $\A_{12}$ was zero to begin with). In this setting the method proposed returns $\widehat{\A}=\mat{0}$. An alternative in this case is to search for an approximation of size $k-1$ or $k+1$, so that $r<\frac{n}{2}$, and the system in Equation~\ref{eq:case2} is underdetermined.
    \end{proof}

\subsubsection{Extraction of the Rational Component}\label{sec:rational_comp}

    The function $\psi(z)=\widehat{\balpha}^{\top}(z\mat{1}-\widehat{\A})^{-1} \widehat{\bbeta}+C$ found in Theorem~\ref{Theorem:maintrm} corresponds to the solution of Theorem~\ref{theorem:aaksymb}. To find the solution to the approximation problem we only need to ``isolate'' the function $g(z)\in \mathcal{R}_k$, \emph{i.e.} the \emph{rational component}. To do this, we study the position of the poles of $\psi(z)$, since the poles of a strictly proper rational function lie in the unit disc (Theorem~\ref{theorem:Kronecker}). As noted before, we parametrized $\psi(z)$ so that its poles correspond to the eigenvalues of $\widehat{A}$. After a change of basis (detailed in the Paragraph~\ref{blockdiag}), we can rewrite $\widehat{\A}$ in block-diagonal form: 
    \begin{equation}
            \widehat{\A}= \begin{pmatrix}
                      \widehat{\A}_+          & \mat{0}                \\
                      \mat{0}                 & \widehat{\A}_-     
                      \end{pmatrix}
    \end{equation}
    where the modulus of the eigenvalues of $\widehat{\A}_+$ (resp. $\widehat{\A}_-$) is smaller (resp. greater) than one. We apply the same change of coordinates on $\widehat{\balpha}$ and $\widehat{\bbeta}$.
    
    To conclude the study of the eigenvalues of $\widehat{\A}$, we need the following auxiliary result from Ostrowski~\cite{ostrowski}. A proof of this theorem can be found in~\cite{inertiaproof}.
    \begin{theorem}[\cite{ostrowski}]\label{theorem:inertia}
        Let $|\Sigma|=1$, and let $\mP$ be a solution to the fixed point equation $X-\A X\A^{\top}=\bbeta\bbeta^{\top}$ for the WFA $A=\langle \balpha, \A, \bbeta \rangle$. If $A$ is reachable, then:
        \begin{itemize}
            \item The number of eigenvalues $\lambda$ of $\A$ such that $|\lambda|<1$ is equal to the number of positive eigenvalues of $\mP$.
            \item The number of eigenvalues $\lambda$ of $\A$ such that $|\lambda|>1$ is equal to the number of negative eigenvalues of $\mP$.
        \end{itemize}
    \end{theorem}
    
    We can finally find the rational component of the function $\psi(z)$, \emph{i.e.} the function $g(z)$ from Corollary~\ref{corollary:stable} necessary to solve that approximate minimization problem.
    \begin{theorem}
        Let $\widehat{\A}_+, \widehat{\balpha}_+, \widehat{\bbeta}_+$ be as in Theorem~\ref{Theorem:maintrm}. The rational component of $\psi(z)$ is the function $g(z)= \widehat{\balpha}_+^{\top}(z\mat{1}-\widehat{\A}_+)^{-1}\widehat{\bbeta}_+$.
    \end{theorem}
    \begin{proof}
        Clearly $\psi(z)=g(z)+ l(z)$, with $l(z)=\widehat{\balpha}_-^{\top}(z\mat{1}-\widehat{\A}_-)^{-1}\widehat{\bbeta}_-$, $l \in \mathcal{H}^{\infty}$. To conclude the proof we need to show that $g(z)$ has $k$ poles inside the unit disc, and therefore has rank $k$. We do this by studying the position of the eigenvalues of $\widehat{\A}_+$. 
        
        Since $E$ is minimal, by definition $\widehat{A}$ is reachable, so we can use Theorem~\ref{theorem:inertia} and solve the problem by directly examining the eigenvalues of $-\mat{\Sigma} \mat{R}$. From the proof of Theorem~\ref{Theorem:maintrm} we have $-\boldsymbol{\Sigma} \mat{R}=\boldsymbol{\Sigma}(\boldsymbol{\Sigma}^2-\sigma^2_k\mat{1})$, where $\boldsymbol\Sigma$ is the diagonal matrix having as elements the singular numbers of $H$ different from $\sigma_k$. It follows that $-\boldsymbol{\Sigma} \mat{R}$ has only $k$ strictly positive eigenvalues, and $\widehat{\A}$ has $k$ eigenvalues with modulus smaller than $1$. Thus, $\widehat{\A}_+$ has $k$ eigenvalues, corresponding to the poles of $g(z)$.
    \end{proof}
    
    \paragraph{Block Diagonalization.}\label{blockdiag}
    In this paragraph we detail the technical steps necessary to rewrite $\widehat{\A}$ in block-diagonal form. The problem of computing the Jordan form of a matrix is ill-conditioned, hence it is not suitable for our algorithmic purposes. Instead, we compute the Schur decomposition, \emph{i.e.} we find an orthogonal matrix $\mat{U}$ such that $\mat{U}^{\top}\widehat{\A}\mat{U}$ is upper triangular, with the eigenvalues of $\widehat{\A}$ on the diagonal. We obtain:
    \begin{equation}\label{eq:T}
        \mat{T}=\mat{U}^{\top}\widehat{\A}\mat{U}= \begin{pmatrix}
                      \widehat{\A}_{+}     & \widehat{\A}_{12}               \\
                      \mat{0}                  & \widehat{\A}_{-}     
                      \end{pmatrix}
    \end{equation}
    where the eigenvalues are arranged in increasing order of modulus, and the modulus of those in $\widehat{\A}_{+}$ (resp. $\widehat{\A}_{-}$) is smaller (resp. greater) than one. To transform this upper triangular matrix into a block-diagonal one, we use the following result.
    \begin{theorem}[\cite{Roth}]\label{bart}
        Let $\mT$ be the matrix defined in Equation~\ref{eq:T}. The matrix $\mat{X}$ is a solution of the equation $\widehat{\A}_{+}\mat{X}- \mat{X}\widehat{\A}_{-} +\widehat{\A}_{12}=\mat{0}$ if and only if the matrices
        \begin{equation}
            \mat{M}=\begin{pmatrix}
                      \mat{1}     & \mat{X}               \\
                      \mat{0}                  & \mat{1}     
                      \end{pmatrix}, \quad \text{and} \quad \mat{M}^{-1}=\begin{pmatrix}
                      \mat{1}     & -\mat{X}               \\
                      \mat{0}                  & \mat{1} \end{pmatrix}
        \end{equation}
        satisfy:
    \begin{equation}
        \mat{M}^{-1}\mat{T}\mat{M}=\begin{pmatrix}
                      \widehat{\A}_{+}     & \mat{0}               \\
                      \mat{0}                  & \widehat{\A}_{-}     
                      \end{pmatrix},
    \end{equation}
    where $\mT$ is the matrix defined in Equation~\ref{eq:T}.
    \end{theorem}
    Setting $\boldsymbol{\Gamma}=\begin{pmatrix}\mat{1}_k & \mat{0} \end{pmatrix}$ we can now derive the rational component of the WFA:
    \begin{align}
        &\widehat{\A}_+= \boldsymbol{\Gamma}\mat{M}^{-1}\mat{U}^{\top}\widehat{\A}\mat{U}\mat{M}\boldsymbol{\Gamma}^{\top}\\
        &\widehat{\balpha}_+= \boldsymbol{\Gamma} \mat{M}^{\top}\mat{U}^{\top}\widehat{\balpha}\\    
        &\widehat{\bbeta}_+= \boldsymbol{\Gamma}\mat{M}^{-1}\mat{U}^{\top}\widehat{\bbeta}.
    \end{align}

\subsubsection{Solution to the Approximation Problem}    
    
    In the previous sections, we have derived the rational function $g(z)$ corresponding to the symbol of $G$, the operator that solves Theorem~\ref{theorem:aakop}. To find the solution to the approximation problem we only need to find the parameters of $\widehat{A}_k$, the optimal approximating WFA. These are directly revealed by the expression of $g(z)$, thanks to the way we parametrized the functions.
    \begin{theorem}
         Let $A= \langle \balpha, \A, \bbeta \rangle$ be a minimal WFA with $n$ states over a one-letter alphabet. Let $A$ be in its SVA form. The optimal spectral-norm approximation of rank $k$ is given by the WFA $\widehat{A}_k= \langle \widehat{\balpha}_+, \widehat{\A}_+, \widehat{\bbeta}_+ \rangle$.
    \end{theorem}
    \begin{proof}
        From Corollary~\ref{corollary:stable} we know that $g(z)$ is the rational function associated with the Hankel matrix of the best approximation. Given the correspondence between the Fourier coefficients of $g(z)$ and the entries of the matrix (Remark~\ref{remark}), we have:
        \begin{equation}
            g(z)= \widehat{\balpha}_+^{\top}(z\mat{1}-\widehat{\A}_+)^{-1}\widehat{\bbeta}_+ = \sum_{k\geq 0} \widehat{\balpha}_+^{\top}\widehat{\A}_+^k \widehat{\bbeta}_+ z^{-k-1} =\sum_{k\geq 0}\bar{f}(k) z^{-k-1}
        \end{equation}
    where $\bar{f}:\Sigma^* \rightarrow \R$ is the function computed by $\widehat{A}_k$ and $\widehat{\balpha}_+, \widehat{\A}_+, \widehat{\bbeta}_+$ are the parameters.
    \end{proof}
 
\subsection{Error Analysis}
    
    The theoretical foundations of AAK theory guarantee that the construction detailed in Section~\ref{method} produces the rank $k$ optimal spectral-norm approximation of a WFA satisfying our assumptions, and the singular number $\sigma_k$ provides the exact error.
    
    Similarly to the case of SVA truncation~\cite{Balle19}, due to the ordering of the singular numbers, the error decreases when $k$ increases, meaning that allowing $\widehat{A}_k$ to have more states guarantees a better approximation of $A$. We remark that, while the solution we propose is optimal in the spectral norm, the same is not necessarily true in other norms. Nonetheless, we have the following bound between $\ell^2$-norm and spectral-norm (proof in Appendix~\ref{appendix:proofs}). 
  
    \begin{theorem}\label{l2bound}
        Let $A$ be a minimal WFA computing $f:\Sigma^* \rightarrow \R$, with matrix $\H$. Let $\widehat{A}_k$ be its optimal spectral-norm approximation, computing $g:\Sigma^* \rightarrow \R$, with matrix $\mat{G}$. Then:
        \begin{equation}
            \norm{f-g}_{\ell^2} \leq \norm{\H - \mat{G}} = \sigma_k.
        \end{equation}
    \end{theorem}
    
    \begin{proof}
        Let $\mat{e}_0=\begin{pmatrix} 1 & 0 & \cdots \end{pmatrix}^{\top}$, $f:\Sigma^*\rightarrow \R$, $g:\Sigma^*\rightarrow \R$ with Hankel matrices $\H$ and $\mat{G}$, respectively. We have:
        \begin{equation*}
            \norm{f-g}_{\ell^2}=\left(\sum_{n=0}^{\infty}|f_n-g_n|^2 \right)^{1/2}=\norm{(\H-\mat{G}) \mat{e}_0}_{\ell^2} \leq \sup_{\norm{\mat{x}}_{\ell^2}=1}\norm{(\H-\mat{G})\mat{x}}_{\ell^2}=\norm{\H-\mat{G}}=\sigma_k 
        \end{equation*}
        where the second equation follows by definition and by observing that matrix difference is computed entry-wise. 
    \end{proof}

\section{Algorithm}\label{alg}

   In this section we describe the algorithm for spectral-norm approximate minimization. The algorithm takes as input a target number of states $k<n$, a minimal WFA $A$ with $\rho(\A)<1$, $\balpha_2 \neq 0$, $n$ states and in SVA form, and its Gramian $\mat{P}$. Note that, in the case of $\balpha_2 = 0$, it is enough to substitute the Steps $4,5,6$ with the analogous from Equation~\ref{eq:balphazero}. The constraints on the WFA $A$ to be minimal and in SVA form are non essential. In fact a WFA with $n$ states can be minimized in time $O(n^3)$~\cite{berstel}, and the SVA computed in $O(n^3)$~\cite{Balle19}. 
    
    Using the results of Theorem~\ref{Theorem:maintrm}, we outline in Algorithm~\ref{alg:approx}, \texttt{AAKapproximation}, the steps necessary to extract the best spectral-norm approximation of a WFA.
    
    \begin{algorithm}[t]
    \caption{\texttt{AAKapproximation}}\label{alg:approx}
        \SetAlgoVlined
        \DontPrintSemicolon
        \SetKwInOut{Input}{input}
        \SetKwInOut{Output}{output}
        \Input{A minimal WFA $A$, with $\balpha_2\neq 0$, $n$ states and in SVA form, \newline its Gramian $\mat{P}$, a target number of states $k<n$}
        \Output{A WFA $\widehat{A}_k$ with $k$ states}
        Let $\balpha_1,\balpha_2,\bbeta_1,\bbeta_2,\A_{11},\A_{12},\A_{22},\mat{\Sigma}$ be the blocks defined in Eq.~\ref{eq:setup}\;
        Let $(\balpha_2^{\top})^+= \frac{\balpha_2}{\balpha_2^{\top}\balpha_2}$, $(\bbeta_2^{\top})^+= \frac{\bbeta_2}{\bbeta_2^{\top}\bbeta_2}$\;
        Let $\mat{R}=\sigma_k^2\mat{1}-\mat{\Sigma}^2$\;
        Let $\widehat{\A}=(\A_{11}- \A_{12}(\bbeta_2^{\top})^+\bbeta_1^{\top})^{-1}$\;
        Let $\widehat{\balpha} = \widehat{\A}^{\top}\mat{R}\A_{12}(\balpha_2^{\top})^{+}$\;
        Let $\widehat{\bbeta} = - \widehat{\A}\A_{12}(\bbeta_2^{\top})^{+}$\;
        Let $\widehat{A}=\langle \widehat{\balpha}, \widehat{\A}, \widehat{\bbeta} \rangle$\;
        Let $\widehat{A}_k \leftarrow$ \texttt{BlockDiagonalize}($\widehat{A}$)\;
        \Return $\widehat{A}_k$ 
    \end{algorithm}

    The algorithm involves a call to Algorithm~\ref{alg:blockd}, \texttt{BlockDiagonalize}. In particular, this corresponds to the steps, outlined in Paragraph~\ref{blockdiag}, necessary to derive the WFA $\widehat{A}_k$ corresponding to the rational function $g(z)$. We remark that Step $2$ in \texttt{BlockDiagonalize} can be performed using the Bartels-Stewart algorithm~\cite{BartelStew}.
    
    \begin{algorithm}[t]
    \caption{\texttt{BlockDiagonalize}}\label{alg:blockd}
        \SetAlgoVlined
        \DontPrintSemicolon
        \SetKwInOut{Input}{input}
        \SetKwInOut{Output}{output}
        \Input{A WFA $\widehat{A}$}
        \Output{A WFA $\widehat{A}_k$ wit $\rho<1$}
        Compute the Schur decomposition of $\widehat{\A}=\mat{U}\mat{T}\mat{U}^{\top}$, where $|T_{11}|\leq |T_{22}| \leq \dots$\;
        Solve $\widehat{\A}_{11}\mat{X}- \mat{X}\widehat{\A}_{22}+\widehat{\A}_{12}=\mat{0}$ for $\mat{X}$\;
        Let $\mat{M}=\begin{pmatrix}
                      \mat{1}     & \mat{X}               \\
                      \mat{0}                  & \mat{1}     
                      \end{pmatrix}$ and $\mat{M}^{-1}=\begin{pmatrix}
                      \mat{1}     & -\mat{X}               \\
                      \mat{0}                  & \mat{1} \end{pmatrix}$\;
        Let $\boldsymbol{\Gamma}=\begin{pmatrix}\mat{1}_k & \mat{0} \end{pmatrix}$\;
        Let $\widehat{\A}_+= \boldsymbol{\Gamma}\mat{M}^{-1}\mat{U}^{\top}\widehat{\A}\mat{U}\mat{M}\boldsymbol{\Gamma}^{\top}$\;
        Let $\widehat{\balpha}_+= \boldsymbol{\Gamma} \mat{M}^{\top}\mat{U}^{\top}\widehat{\balpha}$\;
        Let $\widehat{\bbeta}_+= \boldsymbol{\Gamma}\mat{M}^{-1}\mat{U}^{\top}\widehat{\bbeta}$\;
        Let $\widehat{A}_k=\langle \widehat{\balpha}_+, \widehat{\A}_+, \widehat{\bbeta}_+ \rangle$\;
        \Return $\widehat{A}_k$
    \end{algorithm}

    To compute the computational cost we recall the following facts~\cite{Computationalcost}:
    \begin{itemize}
        \item The product of two $n\times n$ matrices can be computed in time $O(n^3)$ using a standard iterative algorithm, but can be reduced to $O(n^{\omega})$ with $\omega<2.4$.
        \item The inversion of a $n\times n$ matrix can be computed in time $O(n^3)$ using Gauss-Jordan elimination, but can be reduced to $O(n^{\omega})$ with $\omega<2.4$.
        \item The computation of the Schur decomposition of a $n\times n$ matrix can be done with a two-step algorithm, where each step takes $O(n^3)$, using the Hessenberg form of the matrix. 
        \item The Bartels-Stewart algorithm applied to upper triangular matrices to find a matrix $m\times n$ takes $O(mn^2+nm^2)$.
    \end{itemize}
    The running time of \texttt{BlockDiagonalize} with input a WFA $\widehat{A}$ with $(n-r)$ states is thus in $O((n-r)^3)$, where $r$ is the multiplicity of the singular value considered. The running time of \texttt{AAKapproximation} for an input WFA $\widehat{A}$ with $n$ states is in $O((n-r)^3)$.

\section{Related Work}
    The study of approximate minimization for WFAs is very recent, and only a few works have been published on the subject. In~\cite{Balle15,Balle19} the authors present an approximate minimization technique using a canonical expression for WFAs, and provide bounds on the error in the $\ell^2$ norm. The result is supported by strong theoretical guarantees, but it is not optimal in any norm. An extension of this method to the case of Weighted Tree Automata can be found in~\cite{ballerabusseau}. A similar problem is addressed in~\cite{kulesza2015}, with less general results. In~\cite{kulesza14}, the authors connect spectral learning to the approximate minimization problem of a small class of Hidden Markov models, bounding the error in terms of the total variation distance.
    
    The control theory community has largely studied approximate minimization in the context of linear time-invariant systems, and several methods have been proposed~\cite{antoulas}. A parallel can be drown between those results and ours, by noting that the impulse response of a discrete time-invariant Single-Input-Single-Output SISO system can be parametrized as a WFA over a one-letter alphabet. In~\cite{Glover} Glover presents a state-space solution for the case of continuous Multi-Input-Multi-Output MIMO systems. Glover's method led to a widespread application of these results, thanks to its computational and theoretical simplicity. This stems from the structure of the Lyapunov equations for continuous systems. It is however not the case for discrete control systems, where the Lyapunov equations have a quadratic form. As noted in~\cite{discreteH}, there is not a simple closed form formula for the state space solution of a discrete system. Thus, most of the results for the discrete case work with a suboptimal version of the problem, with restrictions on the multiplicity of the singular values~\cite{balldiscrete, Al-Hussari, Ionescu}. A solution for the SISO case can be found, without additional assumptions, using a polynomial approach, but it does not provide an explicit representation of the state space nor it generalizes to the MIMO setting. The first to actually extend Glover results to the discrete case is Gu, who provides an elegant solution for the MIMO discrete problem~\cite{gu}. Glover and Gu's solutions rely on embedding the initial system into an extension of it, the \emph{all-pass system}, equivalent to the WFA $E$ in our method. Part of our contribution is the adaptation of some of the control theory tools to our setting.

\section{Conclusion}
    
    In this paper we applied the AAK theory for Hankel operators and complex functions to the framework of WFAs in order to construct the best possible approximation to an automaton given a bound on the size. We provide an algorithm to find the parameters of the best WFA approximation in the spectral norm, and bounds on the error. Our method applies to real WFAs $A=\langle \balpha, \A, \bbeta \rangle$, defined over a one-letter alphabet, with $\rho(\A)<1$. While this setting is certainly restricted, we believe that this work constitutes a first fundamental step towards optimal approximation. Furthermore, the use of AAK techniques has proven to be very fruitful in related areas like control theory; we think that automata theory can also benefit from it. The use of such methods can help deepen the understanding of the behaviour of rational functions. This paper highlights and strengthens the interesting connections between functional analysis, automata theory and control theory, unifying tools from different domains in one formalism.
    
    A compelling direction for future work is to extend our results to the multi-letter case. The work of Adamyan, Arov and Krein provides us with a powerful theory connecting sequences to the study of complex functions. We note that, unfortunately, this approach cannot be directly generalized to the multi-letter case because of the non-commutative nature of the monoid considered. Extending this work would require adapting standard harmonic analysis results to the non-abelian case. A recent line of work in functional analysis is centered around extending this theory to the case of non-commutative operators, and in~\cite{popescu} a non-commutative version of the AAK theorem is presented. However, those results are non-constructive, making this direction, already challenging, even harder to pursue.
    
    \section*{Acknowledgments}
    
    This research has been supported by NSERC Canada (C. Lacroce, P. Panangaden, D. Precup) and Canada CIFAR AI chairs program (Guillaume Rabusseau). The authors would like to thank Tianyu Li, Harsh Satija and Alessandro Sordoni for feedback on earlier drafts of this work, Gheorghe Comanici for a detailed review, and Maxime Wabartha for fruitful discussions and for comments on the proofs.

\bibliography{bibliography}
\newpage

\appendix
\section{Hankel Operators}\label{appendix:fun_an}

    For more details on the content of this section we refer the reader to~\cite{Nikolski}.
    We recall the first definition of Hankel operator.
    
    \begin{definition}\label{Hankel1}
        A \textbf{Hankel operator} is a mapping $H: \ell^2 \rightarrow \ell^2$ with matrix $\H=\{\alpha_{j+k}\}_{j,k \geq 0}$. In other words, given $a=\{a_n\}_{n \geq 0} \in \ell^2$, we have $H(a) = b$, where $b=\{b_n\}_{n \geq 0}$ is defined by:
        \begin{equation}
            b_k=\sum_{j\geq 0}\alpha_{j+k}a_j, \quad k\geq 0.
        \end{equation}
    \end{definition}
    
    This property on the Hankel matrix can be rephrased as an operator identity. Defining the shift operator by $S(x_0,x_1,\dots)=(0,x_0,x_1,\dots)$ and denoting its left inverse by $S^*=(y_0,y_1,\dots)=(y_1,y_2,\dots)$, we have that $H$ is a Hankel operator if and only if:
    \begin{equation}
        HS=S^*H.
    \end{equation}
    
    The correspondence between Definition~\ref{Hankel1} and Definition~\ref{Hankel2} can be easily made explicit. First, we note that using the isomorphism $\phi \mapsto (\widehat{\phi}(n))_{n \in \Z}$, introduced with Fourier series, we can identify $\mathcal{H}^2$ with $\ell^2$. Moreover, the operator of multiplication by $z$ acts as right shift $S$ on the space of Fourier coefficients, in fact $\widehat{(z\phi)}(n)=\widehat{\phi}(n-1)$. Analogously, the left inverse $S^*$ corresponds to the truncated multiplication operator. Now, let $\mathcal{B}_1=\{z^k\}_{k\geq 0}$ and $\mathcal{B}_2=\{z^j\}_{j<0}$ be bases for $\mathcal{H}^2$ and $\mathcal{H}^2_-$, respectively, and let
    \begin{equation}
        Iz^n=z^{-n-1}
    \end{equation}
    be the involution on $\mathcal{L}^2(\mathbb{T})$. Note that $I\mathcal{H}^2=\mathcal{H}^2_-$.
    
    Let $\overline{H}: \ell^2 \rightarrow \ell^2$ be a Hankel operator with matrix $\overline{\H}$. Using the bases $\mathcal{B}_1$, $\mathcal{B}_2$ and the Fourier identification map, we can obtain an operator acting between Hardy spaces. Following this interpretation, the operator $H=I\overline{H}:\mathcal{H}^2 \rightarrow \mathcal{H}^2_-$ has matrix $\overline{\H}$ with respect to $\mathcal{B}_1$, $\mathcal{B}_2$, and satisfies:
        \begin{equation}\label{hankel_condition}
            HS=\mathbb{P}_-SH.
        \end{equation}
    In particular, $\overline{H}S=S^*\overline{H}$ if and only if $HS=\mathbb{P}_-SH$.
    It is now easy to see that the characterization of the Hankel operator given in Definition~\ref{Hankel2} satisfies Equation~\ref{hankel_condition}.

    
   The following theorem, due to Nehari~\cite{Nehari}, is of great importance as it highlights a correspondence between bounded Hankel operators and functions in $\mathcal{L}^{\infty}(\mathbb{T})$. 

    \begin{theorem}[\cite{Nehari}]\label{thm:nehari1}
        A Hankel operator $H: \ell^2 \rightarrow \ell^2$ with matrix $\H(j,k)=\{\alpha_{j+k}\}_{j,k \geq 0}$ is bounded on $\ell^2$ if and only if there exists a function $\psi \in \mathcal{L}^{\infty}(\mathbb{T})$ such that
        \begin{equation}
            \alpha_m=\widehat{\psi}(m), \quad m\geq 0.
        \end{equation}
        In this case:
        \begin{equation}
             \norm{H} =\inf\{\norm{\psi}_{\infty}:\widehat{\psi}(n)=\widehat{\phi}(n), \, n \geq 0\}.
        \end{equation}
        Where $\widehat{\psi}(n)$ is the $n$-th Fourier coefficient of $\psi$.
    \end{theorem}
    
    We can now reformulate the theorem using the characterization of Hankel operators in Hardy spaces.
    
    \begin{theorem}[\cite{Nehari}]
        Let $\phi \in  \mathcal{L}^2(\mathbb{T})$ be a symbol of the Hankel operator on Hardy spaces $H_{\phi}:\mathcal{H}^2 \rightarrow \mathcal{H}^2_-$. The following are equivalent:
        \begin{itemize}
            \item[(1)] $H_{\phi}$ is bounded on $\mathcal{H}^2$,
            \item[(2)] there exists $\psi \in \mathcal{L}^{\infty}(\mathbb{T})$ such that $\widehat{\psi}(m)=\widehat{\phi}(m)$ for all $m<0$.
        \end{itemize}
        If the conditions above are satisfied, then:
        \begin{equation}\label{eq:nehari}
            \norm{H_{\phi}}=\inf\{\norm{\psi}_{\infty}:\widehat{\psi}(m)=\widehat{\phi}(m), \, m<0\},
        \end{equation}
        or equivalently:
        \begin{equation}\label{eq:nehari2}
            \norm{H_{\phi}}=\inf_{f(z)\in \mathcal{H}^{\infty}}\norm{\phi(z) - f(z)}_{\infty}.
        \end{equation}
        \end{theorem}
    Nehari's Theorem is at the core of the proof of Corollary~\ref{corollary:stable}.
       
    \begin{proof}[\textbf{Proof of Corollary~\ref{corollary:stable}}]
    
        Let $H_{\phi}$ be a Hankel operator with symbol $\phi(z) \in  \mathcal{L}^{\infty}(\mathbb{T})$ and matrix $\H$. Let $\psi(z)=g(z)+l(z)\in \mathcal{H}^{\infty}_k$ be the solution of Equation~\ref{eqsymbol}. We have:
        \begin{align}
            \norm{H_{\phi}-H_{\psi}} &= \norm{H_{\phi-{\psi}}}\\
                                &=\norm{H_{\sigma_k\eta^{-}_k(z)/\xi^{+}_k(z)}} \\
                                &\leq \sigma_k \norm{\eta^{-}_k(z)/\xi^{+}_k(z)}_{\infty}=\sigma_k
        \end{align}
        where first we used Corollary~\ref{corollary:unimodular} and then Equation~\ref{eq:nehari2}.
        Now, using the definition of Hankel operator, we have:
        \begin{equation}
            \norm{H_{\phi}-H_{\psi}}=\norm{H_{\phi}-H_g}= \norm{\H-\mat{G}}\leq \sigma_k.
        \end{equation}
        Since $\norm{\H-\mat{G}} \geq \sigma_k$ (from Eckart-Young theorem~\cite{Eckart}), it follows that $\norm{\H-\mat{G}}=\sigma_k$. Note that $\mat{G}$ has rank $k$, as required, because $g\in\mathcal{R}_k$(Theorem~\ref{theorem:Kronecker}).
    \end{proof}

\section{Proofs from Section~\ref{section3}}\label{appendix:proofs}

    \paragraph*{Proof of Theorem~\ref{Theorem:allpass}.} In order to prove Theorem~\ref{Theorem:allpass} we need an auxiliary lemma. These are the analogous of a control theory result, rephrased in terms of WFAs. The original theorem and lemma, together with the corresponding proofs, can be found in~\cite{discreteH}. Hence, we only provide a sketch of the proofs.
    
    \begin{lemma}[~\cite{discreteH}]~\label{lemma:T}
        Let $E=\langle \balpha_e, \A_e, \bbeta_e \rangle$ be a minimal WFA. Let $e(z)= \balpha_e^{\top}(z\mat{1}-\A_e)^{-1} \bbeta_e - C$, if $\sigma_k^{-1}e(z)$ is unimodular, then there exist a unique invertible symmetric matrix $\mT$ satisfying:
        \begin{itemize}
            \item[(a)] $\A_e^{\top}\mT\bbeta_e=\balpha_e C$
            \item[(b)] $\sigma_k^2\balpha_e^{\top}\mT^{-1}\A_e^{\top}=C\bbeta_e^{\top}$
            \item[(c)] $\A_e^{\top}\mT\A_e-C^{-1}\A_e^{\top}\mT\bbeta_e\balpha_e^{\top}=\mT$
        \end{itemize}
    \end{lemma}
    \begin{proof}
    
        Since $\sigma_k^{-1}e(z)$ is unimodular, we have that:
        \begin{equation}
            e(z)e^*(\bar{z}^{-1})=\sigma_k^2\mat{1}    
        \end{equation}
        where we denote with $e^*$ the adjoint function. From the equation above, we obtain:
        \begin{align}
            e^*(\bar{z}^{-1})&=\sigma_k^2e^{-1}(z)=\sigma_k^2(-C+\balpha_e^{\top}(z\mat{1}-\A_e)^{-1}\bbeta_e)^{-1}\\
                             &=-\sigma_k^2C^{-1}-\sigma_k^2 C^{-1}\balpha_e^{\top}((z\mat{1}-(\A_e+C^{-1}\bbeta_e\balpha_e))^{-1}\bbeta_eC^{-1}\label{eq:inv}
        \end{align}
        where we used the matrix inversion lemma. On the other hand we have:
        \begin{align}
             e^*(\bar{z}^{-1})&=-C+\bbeta_e^{\top}(z^{-1}\mat{1}-\A_e^{\top})^{-1}\balpha_e\\
                             &= -C+\bbeta_e^{\top}(-\A_e^{-\top}(\mat{1}-z\A_e^{\top})+\A_e^{-\top})(\mat{1}-z\A_e^{\top})^{-1}\balpha_e\\                             &=-(C-\bbeta_e^{\top}\A_e^{-\top}\balpha_e)- \bbeta_e^{\top}\A_e^{-\top}(z\mat{1}-\A_e^{-\top})^{-1}\A_e^{-\top}\balpha_e \label{eq:conj}
        \end{align}
        where we used again the matrix inversion lemma before grouping the terms. If the quantities in Equation~\ref{eq:inv} and Equation~\ref{eq:conj} have to be equal, we need their constant term to be the same. Then, we want the $\mathcal{H}^{\infty}_-$-components to correspond, so we consider the corresponding Hankel matrices. It is easy to see that we can once again associate the coefficients of these complex functions to the parameters of a WFA. From the minimality of $E$ we obtain:
        \begin{equation}
            \begin{cases}
                \sigma_k^2C^{-1}\balpha_e^{\top} =\bbeta_e^{\top}\A_e^{-\top}\mT \\
                \A_e+C^{-1}\bbeta_e\balpha_e= \mT^{-1}\A_e^{-\top}\mT\\
                \bbeta_eC^{-1}=\mT^{-1}\A_e^{-\top}\balpha_e
            \end{cases}
        \end{equation}
        where $\mT$ is an invertible matrix~\cite{BalleCLQ14}. This system is equivalent to:
        \begin{equation}\label{eq:systemlemma}
            \begin{cases}
                \sigma_k^2\balpha_e^{\top}\mT^{-1}\A_e^{\top}=C\bbeta_e^{\top}\\
                \A_e^{\top}\mT\A_e-C^{-1}\A_e^{\top}\mT\bbeta_e\balpha_e^{\top}=\mT\\
                \A_e^{\top}\mT\bbeta_e=\balpha_e C
            \end{cases}
        \end{equation}
        To conclude the proof it remains to check that $\mT$ is symmetric, and this can be checked by direct computations.
    \end{proof}
    
    \begin{proof}[\textbf{Proof of Theorem~\ref{Theorem:allpass}}]
        This proof follows easily from Lemma~\ref{lemma:T} by setting $\mP=-\sigma^2_k\mT^{-1}$ and $\mQ=-\mT$. We obtain point $(c)$ by direct multiplication. Then, we substitute the last equation in~\ref{eq:systemlemma} into the second one, and we obtain:
        \begin{equation}
            \A_e^{\top}\mT\A_e-\balpha_e\balpha_e^{\top}=\mT    
        \end{equation}
        which verifies point $(b)$ with $\mQ=-\mT$. Point $(a)$ can be obtained analogously combining the first and second equations in~\ref{eq:systemlemma}.
    \end{proof}

\section{Possible Extensions}\label{appendix:relax}
 
\subsection{Relaxing the Spectral Radius Assumption}\label{sec:specradius} 

    It is possible to extend part of our method to WFAs over a one-letter alphabet with $\rho(\A)\neq 1$, but the approximation recovered is not optimal in the spectral norm. 
    
    Let $A=\langle \balpha, \A, \bbeta \rangle$, with $\rho(\A)\neq 1$, be a WFA with $n$ states that we want to minimize. The idea is to block-diagonalize $\A$ like we did in Section~\ref{sec:rational_comp}, and  tackle each component separately. The case of $A_+=\langle \balpha_+, \A_+, \bbeta_+ \rangle$, the component having $\rho(\A)<1$, can be dealt with in the way presented in the previous sections. This means that we can find an optimal spectral-norm approximation of the desired size for $A_+$. Now we can consider the second component, $A_-=\langle \balpha_-, \A_-, \bbeta_- \rangle$. The key idea is to apply the transformation $z^{j-1} \mapsto z^{-j}$ for $j\geq 1$ to the function $\phi''(z)$ associated to $A_-$. Then, the function
    \begin{equation}
        \phi''(z^{-1})= \sum_{k\geq 0} \balpha_-^{\top}\A_-^k z^k \bbeta_- = \balpha_-^{\top}(\mat{1}-z\A_-)^{-1}\bbeta_-
    \end{equation}
    is well defined, as the series converges for $z$ with small enough modulus. Using this transformation we obtain a function with poles inside the unit disc and we can apply the method presented in the paper. An important choice to make is the size of the approximation of $A_-$, as it can influence the quality of the approximation. Analyzing the effects of this parameter on the approximation error constitutes an interesting direction for future work, both in the theoretical and experimental side. We refer the reader to the control theory literature~\cite{Glover}, where some theoretical work has been done to study an analogous approach for continuous time systems and their approximation error.

\subsection{Polynomial method} 

    We remark that Equation~\ref{eq:unimodular} from Corollary~\ref{corollary:unimodular} can be rewritten as 
    \begin{equation}
        \psi(z) = \phi(z) - \frac{H\xi^{+}_k(z)}{\xi^{+}_k(z)},
    \end{equation}
    where $\xi^{+}_k(z)$ is the function in $\mathcal{H}^2$ associated to the vector $\boldsymbol{\xi}_k\in \operatorname{Ker}(\H^*\H-\sigma_k^2\mat{1})$ (and $\psi(z)$ does not depend on the choice of the specific $\boldsymbol{\xi}_k$). There is an alternative way to find the best approximation, particularly useful when the objective is to approximate a finite-rank infinite Hankel matrix with another Hankel matrix, without necessarily extract a WFA.
    We can consider the adjoint operator $H^*$ and its matrix $\H^*$. The singular numbers and singular vectors of $H$ correspond to the eigenvalues and eigenvectors of $\mat{R}=(\H^*\H)^{1/2}$. Hence, it is possible to compute $\sigma_k$ and a corresponding singular vector $\boldsymbol{\xi}_k$. The function $\xi^+_k(z)$ is then obtained following Equation~\ref{eq:hardynotation}. Then, the Hankel matrix $\mat{G}$ that best approximates $\H$ is given by $\mat{G} =\H- \mat{M}$, where $\mat{M}$ is the Hankel matrix having $\frac{H\xi^{+}_k(z)}{\xi^{+}_k(z)}$ as symbol.
    
\end{document}